\documentclass[twocolumn]{article}

\usepackage{amssymb}
\usepackage{amsmath}
\usepackage{graphicx}
\usepackage{subfig}
\usepackage[english]{babel}
\usepackage{dsfont}

\usepackage{float}
\restylefloat{table}
\usepackage{booktabs}

\newtheorem{theorem}{Theorem}

\newtheorem{lemma}{Lemma}

\newenvironment{proof}{\paragraph{Proof:}}{\hfill$\square$}

\usepackage{authblk}

\title{Distance domination, guarding and vertex cover for maximal outerplanar graph}
\author[1]{Santiago Canales}
\author[2]{Gregorio Hern\'andez}
\author[3]{Mafalda Martins\thanks{mafalda.martins@ua.pt}}
\author[3,4]{In\^es Matos}
\affil[1]{Universidade de Aveiro, Portugal}
\affil[2]{Universidad Pontificia Comillas de Madrid, Spain}
\affil[3]{Universidad Polit\'ecnica de Madrid, Spain}
\affil[4]{Universitat Polit\`ecnica de Catalunya, Spain}

\renewcommand\Authands{ and }
\date{}

\begin{document}

\maketitle

\begin{abstract}
This paper discusses a distance guarding concept on triangulation graphs, which can be associated with distance domination and distance vertex cover. We show how these subjects are interconnected and provide tight bounds for any \mbox{$n$-vertex} maximal outerplanar graph: the $2d$-guarding number, $g_{2d}(n) = \lfloor \frac{n}{5} \rfloor$; the $2d$-distance domination number, $\gamma_{2d}(n) = \lfloor \frac{n}{5} \rfloor$; and the $2d$-distance vertex cover number, $\beta_{2d}(n) = \lfloor \frac{n}{4} \rfloor$.
\end{abstract}

\maketitle

\section{Introduction}

Domination, covering and guarding are widely studied subjects in graph theory. Given a graph $G=(V,E)$ a \emph{dominating set} is a set \mbox{$D \subseteq V$} of vertices such that every vertex not in $D$ is adjacent to a vertex in $D$. The \emph{domination number} $\gamma(G)$ is the number of vertices in a smallest dominating set for $G$. A set \mbox{$C \subseteq V$} of vertices is a \emph{vertex cover} if each edge of the graph is incident to at least one vertex of the set.  The \emph{vertex cover number} $\beta(G)$ is the size of a minimum vertex cover. Thus, a dominating set guards the \emph{vertices} of a graph while a vertex cover guards its \emph{edges}. In plane graphs, these concepts differ from the notion of \emph{guarding set} as the latter guards the \emph{faces} of the graph. Let \mbox{$G=(V,E)$} be a plane graph, a guarding set is a set $S \subseteq V$ of vertices such that every face has a vertex in $S$. The \emph{guarding number} $g(G)$ is the number of vertices in a smallest guarding set for $G$.

There are many papers and books about domination and its many variants in graphs, e.g. \cite{Campos13,Haynes98,King10,Matheson96}. In 1975, domination was extended to \emph{distance domination} by Meir and Moon \cite{Meir75}. Given a graph $G$, a set $D \subset V$ of vertices is said to be a \emph{distance \mbox{$k$-dominating} set} if for each vertex \mbox{$u \in V-D$}, \mbox{$dist_G(u,v) \leq k$} for some \mbox{$v \in D$}. The minimum cardinality of a distance \mbox{$k$-dominating} set is said to be the \emph{distance \mbox{$k$-domination} number} of $G$ and is denoted by $\gamma_{k}(G)$ or $\gamma_{kd}(G)$. Note that a classical dominating set is a distance $k$-dominating set at distance 1. In the case of distance domination, there are also some known results concerning bounds for $\gamma_{kd}(G)$, e.g., \cite{Sridharan02,Tian04,Tian09}.  However, if graphs are restricted to triangulations, then we are not aware of known bounds for $\gamma_{kd}(G)$. The distance domination was generalized to \emph{broadcast domination}, by Erwin, when the power of each vertex may vary \cite{Erwin04}. Given a graph \mbox{$G = (V,E)$}, a \emph{broadcast} is a function \mbox{$f : V \rightarrow \mathds{N}_0$}.  The cost of a broadcast $f$ over a set $S$ of $V$ is defined as \mbox{$f(S) = \sum_{v \in S} f(v)$}. Thus, $f(V)$ is the total cost of the broadcast function $f$. A broadcast is \emph{dominating} if for every vertex $v$, there is a vertex $u$ with \mbox{$f(u) > 0$} and \mbox{$d(u, v) \leq f(u)$}, that is, a vertex $u$ with non null broadcast and whose broadcast's power reaches vertex $v$. A dominating broadcast $f$ is \emph{optimal} if $f(V)$ is minimum over all choices of broadcast dominating functions for $G$. The \emph{broadcast domination problem} consists in building this optimal function. Note that, if $f(V)=\{0,1\}$, then the  broadcast domination problem coincides with the problem of finding a minimum dominating set with minimum cardinality. And, if $f(V)=\{0,k\}$, then the broadcast domination problem is the distance \mbox{$k$-dominating} problem. If a broadcast $f$ provides coverage to the edges of $G$ instead of covering its vertices, then we have a generalization of the vertex cover concept \cite{Blair05}. A broadcast $f$ is \emph{covering} if for every edge \mbox{$(x,y) \in E$}, there is a path $P$ in $G$ that includes the edge $(x,y)$ and one end of $P$ must be a vertex $u$, where $f(u)$ is at least the length of $P$. A covering broadcast $f$ is \emph{optimal} if $f(V)$ is minimum over all choices of broadcast covering functions for $G$. Note that, if \mbox{$f(V)=\{0,1\}$}, then the broadcast cover problem coincides with the problem of finding a minimum vertex cover. Regarding the broadcast cover problem when all vertices have the same power (i.e., when \mbox{$f(V)=\{0,k\}$}, for a fixed \mbox{$k \neq 1$}), as far as we know, there are no published results besides \cite{Chen12} where the authors propose a centralized and distributed approximation algorithm to solve it.

The guarding concept on plane graphs has its origin in the study of triangulated terrains, polyhedral surfaces whose faces are triangles and with the property that each vertical line intersects the surface at most by one point or segment. A set of guards covers the surface of a terrain if every point on the terrain is visible from at least one guard in the set. The combinatorial aspects of the terrain guarding problems can be expressed as guarding problems on the plane triangulated graph underlying the terrain. Such graph is called \emph{triangulation graph} (\emph{triangulation}, for short), because is the graph of a triangulation of a set of points in the plane (see Figures \ref{FIG:article-arXiv-1} and \ref{FIG:article-arXiv-2}). In this context of guarding for plane graphs, a set of guards only needs to watch the bounded faces of the graph. There are known bounds on the guarding number of a plane graph, $g(G)$; for example, $g(G) \leq \frac{n}{2}$ for any $n$-vertex plane graph \cite{Bose97}, and $g(G) \leq \frac{n}{3}$ for any triangulation of a polygon \cite{Fisk78}. The triangulation of a polygon is a \emph{maximal outerplanar graph}. A graph is outerplanar if it has a crossing-free embedding in the plane such that all vertices are on the boundary of its outer face (the unbounded face). An outerplanar is maximal outerplanar if it is not possible to add an edge such that the resulting graph is still outerplanar. A maximal outerplanar graph embedded in the plane as mentioned above is an maximal outerplanar graph and corresponds to a triangulation of a polygon. Contrary to the notions of domination and vertex cover on plane graphs that were extended to include their distance versions, the guarding concept was not generalized to its distance version.

In this paper we generalize the guarding concept on plane graphs to its distance guarding version and also formalize the broadcast cover problem when all vertices have the same power, which we call distance k-vertex cover. Furthermore, we analyze these concepts of distance guarding, covering and domination, from a combinatorial point of view, for triangulation graph. We obtain tight bounds for distance versions of guarding, domination and vertex covering for maximal outerplanar graphs.

In the next section we first describe some of the terminology used in this paper, and then discuss the relationship between distance guarding, domination and covering on triangulation graphs. In sections \ref{SEC:Domination_and_DistanceTightUpperBounds} and \ref{SEC:CoveringTightUpperBounds} we study how these three concepts of distance apply to maximal outerplanar graphs. And finally, the paper concludes with section 5 that discusses our results and
future research.

\section{Relationship between distance guarding, distance domination and distance vertex cover on triangulation graphs}
\label{SEC:BasicNotions}

In the following we introduce some of the notation used throughout the text, and then proceed to explain the relationship between the different distance concepts on triangulations. Given a triangulation \mbox{$T=(V,E)$}, we say that a bounded face $T_i$ of $T$ (i.e., a triangle) is \emph{$kd$-visible} from a vertex \mbox{$p \in V$}, if there is a vertex \mbox{$x \in T_i$} such that \mbox{$dist_T(x,p) \leq k-1$}. The \emph{$kd$-visibility region} of a vertex \mbox{$p \in V$} comprises the triangles of $T$ that are \mbox{$kd$-visible} from $p$ (see Fig. \ref{FIG:article-arXiv-1}).

\begin{figure}[!htb]
\centering
\includegraphics[scale=0.63]{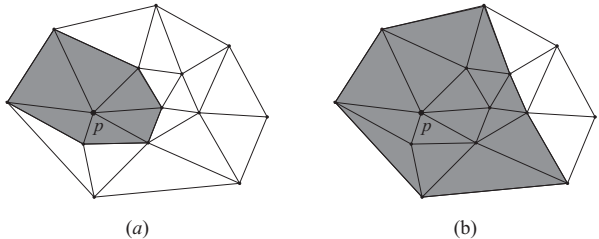}
\caption{The $kd$-visible region of $p$ for: (a) $k=1$; (b) $k=2$.} \label{FIG:article-arXiv-1}
\end{figure}

A \emph{$kd$-guarding set} for $T$ is a subset \mbox{$F \subseteq V$} such that every triangle of $T$ is \mbox{$kd$-visible} from an element of $F$. We designate the elements of $F$ by \emph{$kd$-guards}. The \emph{$kd$-guarding number} $g_{kd}(T)$ is the number of vertices in a smallest \mbox{$kd$-guarding} set for $T$. Note that, to avoid confusion with \emph{multiple guarding} \cite{Belleville09} -- where the typical notation is \mbox{$k$-guarding} -- we will use \mbox{$kd$-guarding}, with an extra ``$d$''. Given a set $S$ of $n$ points, we define
$g_{kd}(S) = max \{g_{kd}(T): T \mbox{ is triangulation with } \mbox{V=S}\}$
and given \mbox{$n \in \mathds{N}$}, $g_{kd}(n) = max \{g_{kd}(S): S \mbox{ is plane point set with } \allowbreak |S|=n\}$.

A \emph{\mbox{$kd$-vertex} cover} for $T$, or \emph{distance \mbox{$k$-vertex} cover} for $T$, is a subset \mbox{$C \subseteq V$} such that for each edge \mbox{$e \in E$} there is a path of length at most $k$, which contains $e$ and a vertex of $C$. The \emph{\mbox{$kd$-vertex} cover number} $\beta_{kd}(T)$ is the number of vertices in a smallest \mbox{$kd$-vertex} cover set for $T$. Given a set $S$ of $n$ points, we define
$\beta_{kd}(S) = max \{\beta_{kd}(S): T \mbox{ is triangulation with } \mbox{V=S}\}$ and given \mbox{$n \in \mathds{N}$}, $\beta_{kd}(n) = max \{\beta_{kd}(S): S \mbox{ is plane point set with } \allowbreak |S|=n\}.$

Finally, as already defined by other authors, a \emph{$kd$-dominating set} for $T$, or \emph{distance $k$-dominating set} for $T$, is a subset \mbox{$D \subset V$} such that each vertex \mbox{$u \in V-D$}, \mbox{$dist_T(u,v) \leq k$} for some \mbox{$v \in D$}. The $kd$-domination number $\gamma_{kd}(T)$ is the number of vertices in a smallest $kd$-dominating set for $T$. Given a set $S$ of $n$ points, we define
$\gamma_{kd}(S) = max \{\gamma_{kd}(T): T \mbox{ is triangulation with } mbox{V=S}\}$
and  given \mbox{$n \in \mathds{N}$},
$\gamma_{kd}(n) \allowbreak = \allowbreak max \{\gamma_{kd}(S): \allowbreak S \allowbreak \mbox{ is plane point set with } \allowbreak |S|=n\}$.

\begin{figure}[!htb]
\centering
\includegraphics[scale=0.5]{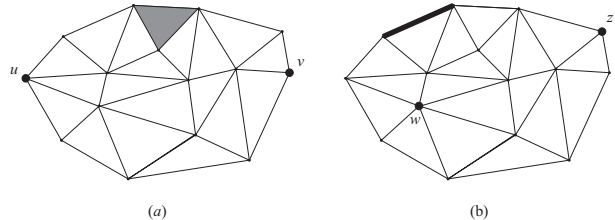}
\caption{(a) $2d$-dominating set for a triangulation $T$; (b) $2d$-guarding set for $T$.} \label{FIG:article-arXiv-2}
\end{figure}

The main goal is to obtain bounds on $g_{kd}(n)$, $\gamma_{kd}(n)$ and $\beta_{kd}(n)$.  We start by showing that the three concepts, $kd$-guarding, $kd$-dominance and $kd$-vertex covering are different. Fig. \ref{FIG:article-arXiv-2} depicts \mbox{$2d$-dominating} and \mbox{$2d$-guarding} sets for a given triangulation $T$. Note that in Fig. \ref{FIG:article-arXiv-2}(a) the set $\{u,v\}$ is $2d$-dominating since the remaining vertices are at distance 1 or 2. However, it is not a \mbox{$2d$-guarding} set because the shaded triangle is not guarded, as its vertices are at distance 2 from $\{u,v\}$. In \ref{FIG:article-arXiv-2}(b) $\{w,z\}$ is a $2d$-guarding set, however it is not a \mbox{$2d$-vertex} cover  since any path between the bold edge and $w$ or $z$ has length at least 3. Therefore, the bold edge is not covered.

Now we are going to establish a relation between $g_{kd}(T)$, $\gamma_{kd}(T)$ and $\beta_{kd}(T)$.

\begin{lemma}
If  $C$ is a $kd$-vertex cover for a triangulation $T$, then $C$ is a $kd$-guarding set and a $kd$-dominating set for $T$.
\end{lemma}

\begin{proof}

If $C$ is a $kd$-vertex cover, then each edge of $T$ has one of its endpoints at distance at most \mbox{$k-1$} from $C$. Thus, any triangle of $T$ has one of its vertices at distance at most \mbox{$k-1$} from $C$, that is, $C$ is a $kd$-guarding set for $T$. Furthermore, all the vertices of $T$ are at a distance of at most $k$ from a vertex of $C$. Therefore $C$ is $kd$-dominant.

\end{proof}

\begin{lemma}
If $C$ is a $kd$-guarding set for a triangulation $T$, then $C$ is a $kd$-dominating set for $T$.
\end{lemma}

\begin{proof}

If $F$ is $kd$-guarding set, then every vertex of $T$ (which belongs to a $kd$-guarded face) is at distance at most $k$ from an element of $F$. Thus, $F$ is $kd$-dominating set for $T$.

\end{proof}

The previous lemmas prove the following result.

\begin{theorem}
\label{Thm:inequalities}
Given a triangulation $T$ the minimum cardinality $g_{kd}(T)$ of any \mbox{$kd$-guarding} set for $T$ verifies

\begin{equation}
\gamma_{kd}(T) \leq g_{kd}(T) \leq \beta_{kd}(T).
\end{equation}

\end{theorem}

Note that the inequalities above can be strict, as we will show for \mbox{$k=2$}. Consider the triangulation $T$ depicted in Fig. \ref{FIG:article-arXiv-3to6}(a). We start by looking for a \mbox{$2d$-dominant} set of minimum cardinality.  The black vertices in Fig. \ref{FIG:article-arXiv-3to6}(b) form a \mbox{$2d$-dominating} set, since each vertex of $T$ is at a distance less than or equal to 2 from a black vertex. Besides, it is clear that the extreme vertices can not be $2d$-dominated by the same vertex, thus any $2d$-dominating set has to have at least two vertices, one to cover each extreme. Consequently, $\gamma_{2d}(T)=2$. But the pair of black vertices is not a $2d$-guarding set because the shaded area is not $2d$-guarded (all the vertices of the shaded triangles are at a distance 2 from the black vertices). Now, we look for a $2d$-guarding set of minimum cardinality. Note that, in Fig. \ref{FIG:article-arXiv-3to6}(d), the gray vertices are a $2d$-guarding set. Each shaded triangle needs one $2d$-guard since they are at distance of 3 and thus every $2d$-guarding set has cardinality at least 3. Therefore, \mbox{$g_{2d}(T)=3$}. Finally, we seek a \mbox{$2d$-vertex} cover. In Fig. \ref{FIG:article-arXiv-3to6}(e), each of the bold edges needs a different vertex to be $2d$-covered, since the distance between each pair of edges is greater than or equal to 3. In this way no single vertex can
simultaneously $2d$-cover two of the bold edges. Thus, \mbox{$\beta_{2d}(T) \geq 4$}. Note that, this example can easily be generalized to any value of $k$.

\begin{figure}[!htb]
  \centering
  \subfloat[]{\includegraphics[width=0.35\textwidth]{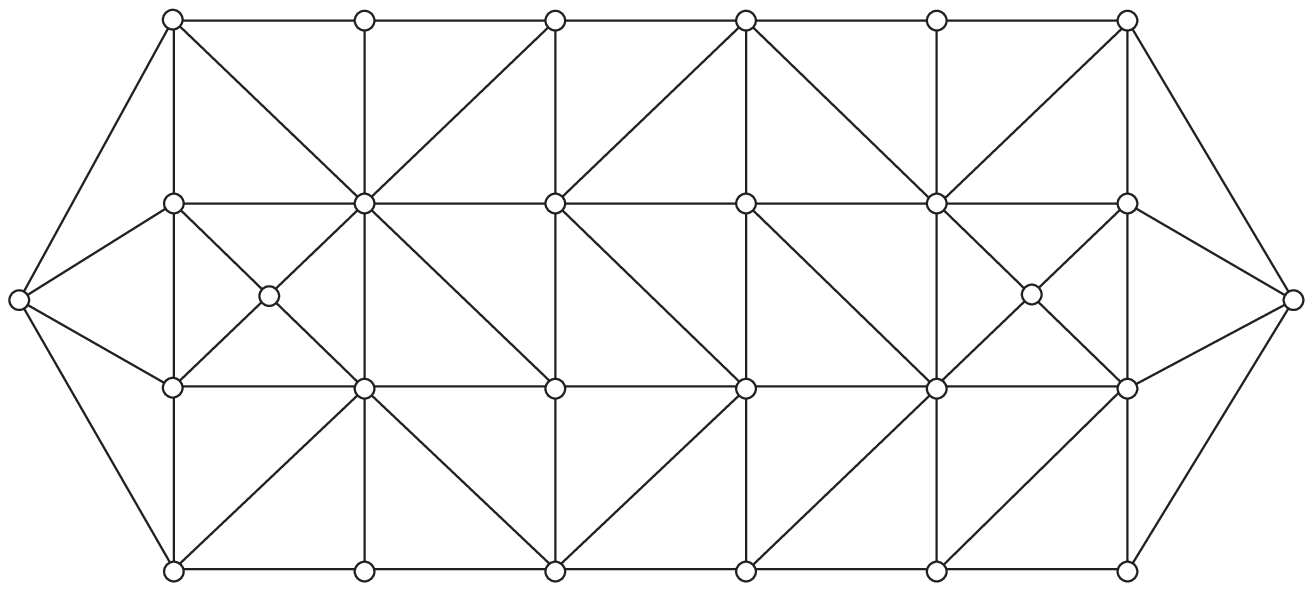}}
  \hspace{0.5cm}
  \subfloat[]{\includegraphics[width=0.35\textwidth]{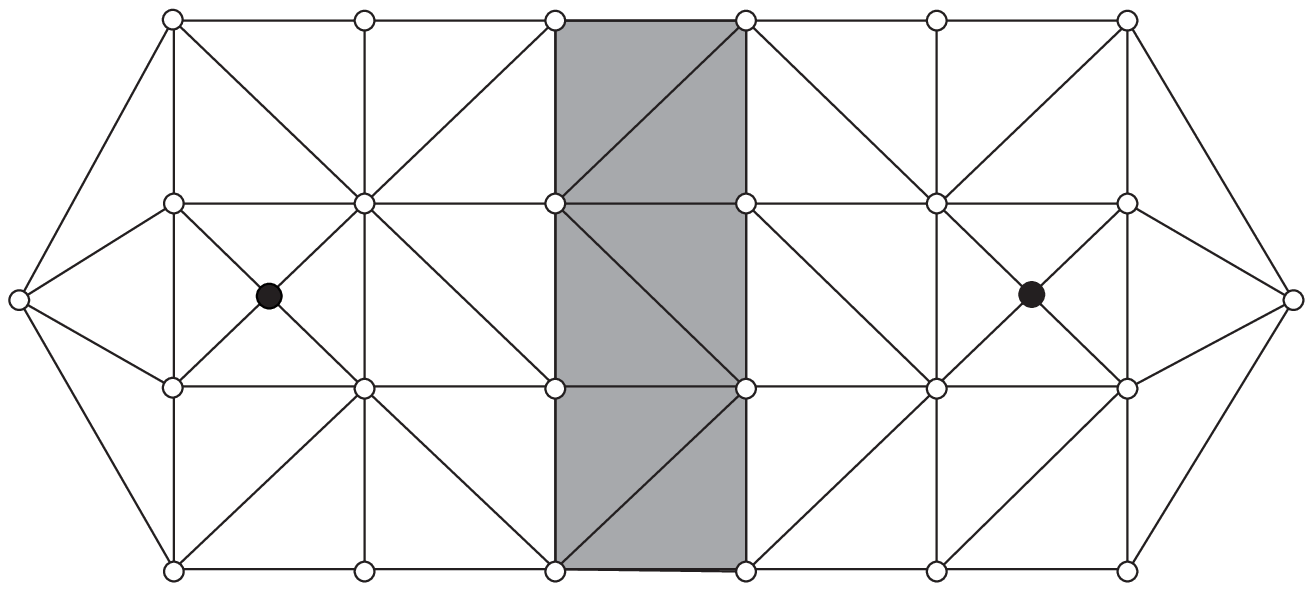}}
  \hspace{0.5cm}
  \subfloat[]{\includegraphics[width=0.35\textwidth]{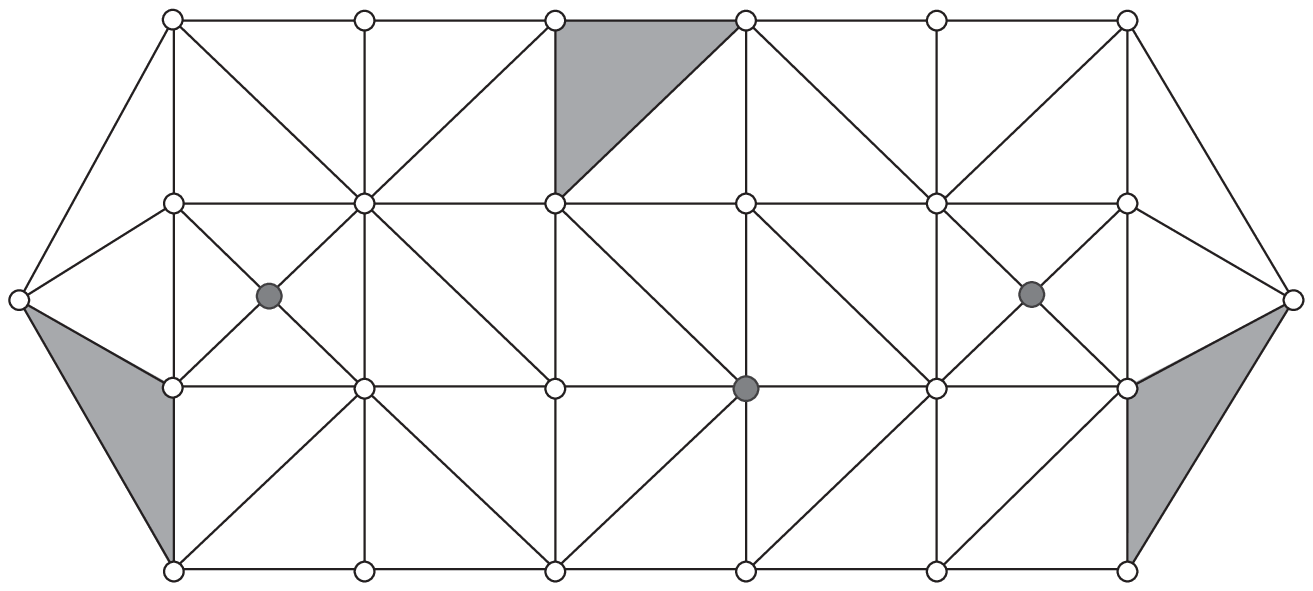}}
  \hspace{0.5cm}
  \subfloat[]{\includegraphics[width=0.35\textwidth]{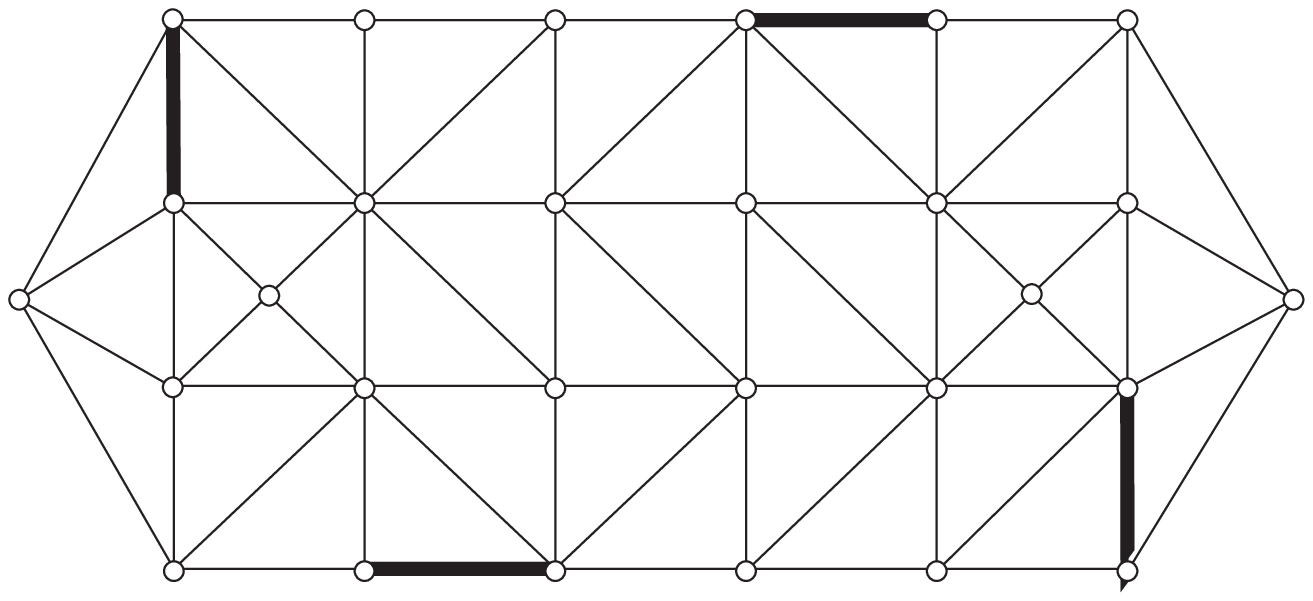}}
  \caption{(a) A triangulation $T$; (b) a $2d$-dominating set for a triangulation $T$ (black vertices); (c) a $2d$-guarding set for $T$ (gray vertices); (d) each of the bold edges needs a different vertex to be $2d$-covered.}
  \label{FIG:article-arXiv-3to6}
\end{figure}

\section{$2d$-guarding and $2d$-domination of maximal outerplanar graphs}
\label{SEC:Domination_and_DistanceTightUpperBounds}

In this section we establish tight bounds for $g_{2d}(n)$ and $\gamma_{2d}(n)$ on a special class of triangulation graphs -- the maximal outerplanar graphs -- which correspond, as stated above, to triangulations of polygons. We call the edges on the exterior face \emph{exterior edges}, otherwise they are \emph{interior edges}. In order to do this, and following the ideas of O'Rourke \cite{O'Rourke83}, we first need to introduce some lemmas.

\begin{lemma}
\label{Lem:f(m-2)}
Suppose that $f(m)$ $2d$-guards are always sufficient to guard any outerplanar maximal graph with $m$ vertices. If $G$ is an arbitrary outerplanar maximal graph with two $2d$-guards placed at any two adjacent of its $m$ vertices, then $f(m-2)$ additional $2d$-guards are sufficient to guard $G$.
\end{lemma}

\begin{proof}

Let $a$ and $b$ be the adjacent vertices at which the $2d$-guards are placed, and $c$ the vertex on the exterior face of $G$ adjacent to $b$. Contract the edges $(a,b)$ and $(b,c)$ of $G$ to produce the outerplanar maximal graph $G^*$ of \mbox{$m-2$} vertices, that is, remove the edges $(a,b)$ and $(b,c)$ and replace them with a new vertex $x$ adjacent to every vertex to which $a$,$b$ and $c$ were adjacent to (see Fig. \ref{FIG:article-arXiv-7}).

 \begin{figure}[!htb]
  \centering
  \includegraphics[scale=0.4]{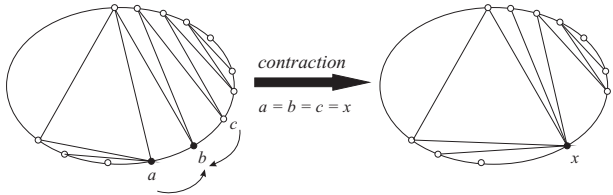}
  \caption{Contraction of the edges $(a,b)$ and $(b,c)$.}
  \label{FIG:article-arXiv-7}
\end{figure}

We know that $f(m-2)$ $2d$-guards are sufficient to guard $G^*$. Suppose that no $2d$-guard is placed at $x$. Then the same $2d$-guarding scheme will guard $G$, since the $2d$-guards placed at $a$ and $b$ guard the triangles with vertices at $a$, $b$ and $c$, and the remaining triangles are guarded by their counterparts counterparts in $G^*$. If a guard is placed at $x$, when the graph is expanded back into $G$, the guard placed at $x$ will be placed at $c$ to assure that $G$ is guarded.

\end{proof}

\begin{lemma}
\label{Lem:f(m-1)Guarding}
Suppose that $f(m)$ $2d$-guards are always sufficient to guard any outerplanar maximal graph with $m$ vertices. If $G$ is an arbitrary outerplanar maximal graph with one $2d$-guard placed at any one of its $m$ vertices, then $f(m-1)$ additional $2d$-guards are sufficient to guard $G$.
\end{lemma}
\begin{proof}

Let $a$ be the vertex where a $2d$-guard is placed and $b$ a vertex on the exterior face of $G$ adjacent to $a$. Contract the edge $(a,b)$ to produce the outerplanar maximal graph $G^*$ of $m-1$ vertices (that is, remove edge $(a,b)$ and replace it with a new vertex $x$ adjacent to every vertex to which $a$ and $b$ were adjacent to). We know that $f(m-1)$ $2d$-guards are sufficient to guard $G^*$. Suppose that no $2d$-guard is placed at $x$. Then the same $2d$-guarding scheme will guard $G$, since the $2d$-guard placed at $a$ covers the triangles with vertices at $a$ and $b$, and the remaining triangles have guarding counterparts in $G^*$. If a guard is placed at $x$, then such guard will be placed at $b$ when the graph is expanded back into $G$. The remaining guards together with $b$ assure that $G$ is $2d$-guarded.

\end{proof}

The next lemma can be easily proven by following the ideas of O'Rourke \cite{O'Rourke83}.

\begin{lemma}
\label{Lem:O'Rourke}
Let $G$ be an outerplanar maximal graph with $n \geq 2k$ vertices. There is an interior edge $e$ in $G$ that partitions $G$ into two components, one of which contains $m=k,k+1,\ldots, 2k-3$ or $2k-2$ exterior edges of $G$.
\end{lemma}

\begin{theorem}
\label{Thm:SufficiencyGuarding}
Every $n$-vertex maximal outerplanar graph, with \mbox{$n \geq 5$}, can be $2d$-guarded by $\lfloor \frac{n}{5} \rfloor$ $2d$-guards. That is, $g_{2d}(n) \leq \lfloor \frac{n}{5} \rfloor$ for all \mbox{$n \geq 5$}.
\end{theorem}

\begin{proof}

For $5 \leq n \leq 11$, the truth of the theorem can be easily established -- the upper bounds are resumed in Table \ref{TAB:guarding}. It should be noted that for \mbox{$n=5$} the $2d$-guard can be placed randomly and for \mbox{$n=6$} it can be placed at any vertex of degree greater than 2 (or one that belongs to an interior edge).

\begin{table}[!htb]
\centering
\small
  \begin{tabular}{ | l | c | c | c | c | c | c | c |  }
    \hline
    $n$ & 5  & 6  & 7  & 8  & 9  & 10  & 11  \\
    \hline \hline
    $g_{2d}(n)$ & 1  & 1  & 1  & 1  & 1  & 2  & 2  \\
    \hline
  \end{tabular}
  \caption{Number of $2d$-guards that suffice to cover a maximal outerplanar graph of $n$ vertices.}
  \label{TAB:guarding}
\end{table}

Assume that \mbox{$n \geq 12$} and that the theorem holds for all \mbox{$n' < n$}. Let $G$ be a triangulation graph with $n$ vertices. The vertices of $G$ are labeled with $0,1,2, \ldots , n$. Lemma \ref{Lem:O'Rourke} guarantees  the existence of an interior edge $e$ (which can be labeled $(0,m)$) that divides $G$ into maximal outerplanar graphs $G_1$ and $G_2$, such that $G_1$ has $m$ exterior edges of $G$ with $6 \leq m \leq 10$. Each value of $m$ will be considered separately.

\begin{enumerate}

\item[(1)] $m=6$. $G_1$ has \mbox{$m+1=7$} exterior edges, thus $G_1$ can be $2d$-guarded with one guard. $G_2$ has $n-5$ exterior edges including $e$, and by induction hypothesis, it can be $2d$-guarded with $\lfloor \frac{n-5}{5} \rfloor = \lfloor \frac{n}{5} \rfloor -1$ guards. Thus $G_1$ and $G_2$ together can be $2d$-guarded by $\lfloor \frac{n}{5} \rfloor$ guards.

\item[(2)] $m=7$. $G_1$ has $m+1=8$ exterior edges, thus $G_1$ can be $2d$-guarded with one guard. $G_2$ has $n-6$ exterior edges including $e$, and by induction hypothesis, it can be $2d$-guarded with $\lfloor \frac{n-6}{5} \rfloor \leq \lfloor \frac{n}{5} \rfloor -1$ guards. Thus $G_1$ and $G_2$ together can be $2d$-guarded by $\lfloor \frac{n}{5} \rfloor$ guards.

\item[(3)] $m=8$. $G_1$ has $m+1=9$ exterior edges, thus $G_1$ can be $2d$-guarded with one guard. $G_2$ has $n-7$ exterior edges including $e$, and by induction hypothesis, it can be $2d$-guarded with $\lfloor \frac{n-7}{5} \rfloor \leq \lfloor \frac{n}{5} \rfloor -1$ guards. Thus $G_1$ and $G_2$ together can be $2d$-guarded by $\lfloor \frac{n}{5} \rfloor$ guards.

\item[(4)] $m=9$.  The presence of any of the internal edges (0,8), (0,7), (0,6), (9,1), (9,2) and (9,3) would violate the minimality of $m$. Thus, the triangle $T$ in $G_1$ that is bounded by $e$ is either (0,5,9) or (0,9,4). Since these are equivalent cases, suppose that $T$ is (0,5,9), see Fig. \ref{FIG:article-arXix-8to10}(a). The pentagon (5,6,7,8,9) can be $2d$-guarded by placing one guard randomly. However, to $2d$-guard the hexagon (0,1,2,3,4,5) we cannot place a $2d$-guard at any vertex. We will consider two separate cases.

\begin{figure}[!htb]
  \centering
  \subfloat[]{\includegraphics[width=0.2\textwidth]{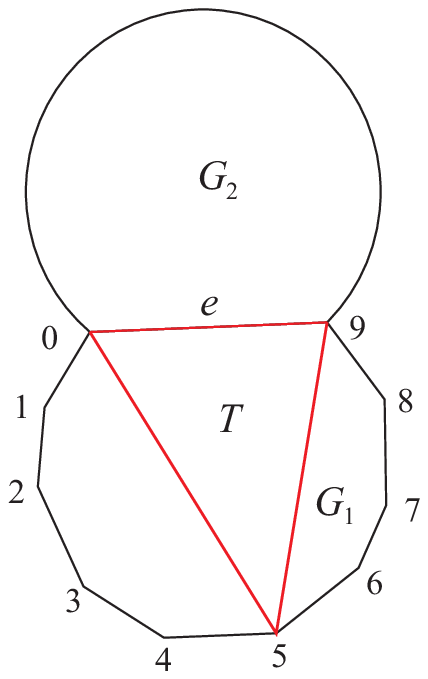}}
  \hspace{0.8cm}
  \subfloat[]{\includegraphics[width=0.2\textwidth]{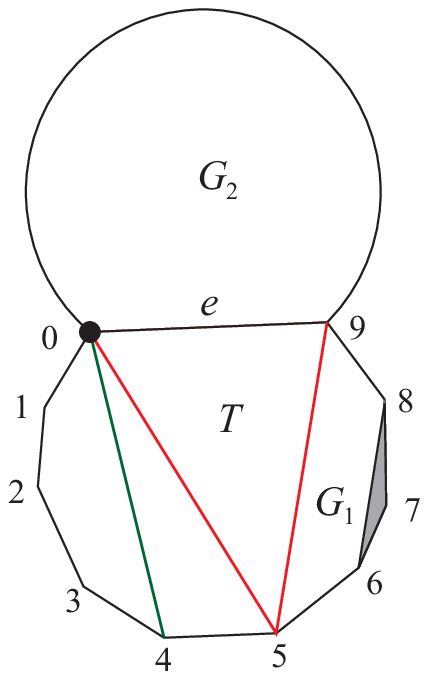}}
  \hspace{0.8cm}
  \subfloat[]{\includegraphics[width=0.2\textwidth]{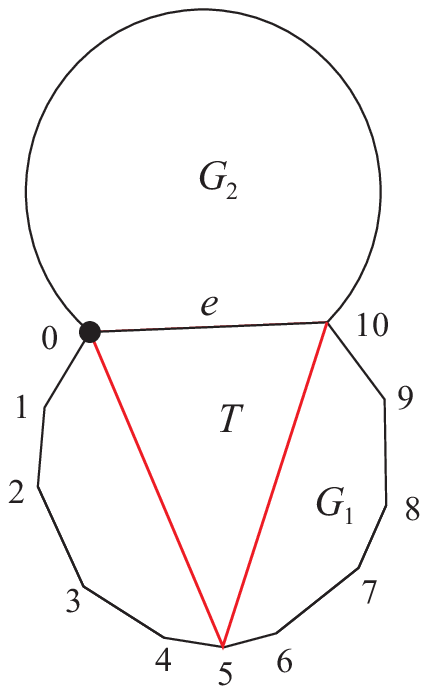}}
  \caption{The interior edge $e$ separates $G$ into two maximal outerplanar graphs $G_1$ and $G_2$: (a) the triangle $T$ in $G_1$ that is bounded by $e$ is (0,5,9); (b) $G_1$ has 10 exterior edges, both the internal edge (0,4) and the triangle (6,7,8) are present; (c) $G_1$ has 11 exterior edges and the triangles (2,3,4) and (6,7,8) are present.}
  \label{FIG:article-arXix-8to10}
\end{figure}

\begin{enumerate}

\item
The internal edge (0,4) is not present. If a guard is placed at vertex 5, then the hexagon (0,1,2,3,4,5) is $2d$-guarded, thus $G_1$ is $2d$-guarded. Since $G_2$ has $n-8$ edges  it can be $2d$-guarded by \mbox{$\lfloor \frac{n-8}{5} \rfloor \leq \lfloor \frac{n}{5} \rfloor -1$} guards by induction hypothesis. This yields a $2d$-guarding of $G$ by $\lfloor \frac{n}{5} \rfloor$ guards.

\item
The internal edge (0,4) is  present. If a $2d$-guard is placed at vertex 0, then $G_1$ is $2d$-guarded unless the triangle (6,7,8) is present in the triangulation (see Fig. \ref{FIG:article-arXix-8to10}(b)). In any case, two $2d$-guards placed at vertices 0 and 9 guard $G_1$. $G_2$ has $n-8$ exterior edges, including $e$. By lemma \ref{Lem:f(m-2)} the two guards placed at vertices 0 and 9 allow the remainder of $G_2$ to be guarded by $f(n-8-2)=f(n-10)$ additional $2d$-guards. Recall that $f(n')$ is the number of $2d$-guards that are always sufficient to guard a maximal outerplanar graph with $n'$ vertices. By the induction hypothesis $f(n')=\lfloor \frac{n'}{5} \rfloor$. Thus, \mbox{$\lfloor \frac{n-10}{5} \rfloor = \lfloor \frac{n}{5} \rfloor - 2$} guards suffice to guard $G_2$. Together with the guards placed at vertices 0 and 9 that $2d$-guard $G_1$, all of $G$ is guarded by $\lfloor \frac{n}{5}\rfloor$ $2d$-guards.

\end{enumerate}

\item[(5)] $m=10$. The presence of any of the internal edges (0,9), (0,8), (0,7), (0,6), (9,1), (9,2), (9,3) and (9,4) would violate the minimality of $m$. Thus, the triangle $T$ in $G_1$ that is bounded by $e$ is (0,5,10) (see Fig. \ref{FIG:article-arXix-8to10}(c)). We will consider two separate cases:

    \begin{enumerate}

        \item
        The vertices 0 and 10  have degree 2 in hexagons (0,1,2,3,4,5) and (5,6,7,8,9,1,0), respectively. Then one $2d$-guard placed at vertex 5 guards $G_1$. By the induction hypothesis $G_2$ can be guarded with \mbox{\mbox{$\lfloor \frac{n-9}{5} \rfloor \leq \lfloor \frac{n}{5}\rfloor-1$}} guards. Thus $G$ can be $2d$-guarded by $\lfloor \frac{n}{5}\rfloor$ guards.

        \item
        The vertex 0 has degree greater than 2 in hexagon (0,1,2,3,4,5). In this case we place a guard at vertex 0 and another guard in one vertex of the hexagon (5,6,7,8,9,1,0) of degree greater than 2. These two guards dominates $G_1$. $G_2$ has $n-9$ vertices. By lemma \ref{Lem:f(m-1)Guarding} the guard placed at vertex 0 permits the remainder of $G_2$ to be $2d$-guarded by \mbox{$f(n-9-1)=f(n-10)$} additional guards, where $f(n')$ is the number of $2d$-guards that are always sufficient to guard a maximal outerplanar graph with $n'$ vertices. By induction hypothesis $f(n')=\lfloor \frac{n'}{5} \rfloor$. Thus, \mbox{$\lfloor \frac{n-10}{5} \rfloor = \lfloor \frac{n}{5}\rfloor-2$} guards suffices to guard $G_2$. Together with the two already allocated to $G_1$, all of $G$ is guarded by $\lfloor \frac{n}{5}\rfloor$ guards.

    \end{enumerate}

\end{enumerate}

\end{proof}

To prove that this upper bound is tight we need to construct a maximal outerplanar graph $G$ of order $n$ such that \mbox{$g_{2d}(G) \geq \lfloor \frac{n}{5} \rfloor$}. Fig. \ref{FIG:article-arXiv-11} shows a maximal outerplanar graph $G$ for which $\gamma_{2d}(G)=\frac{n}{5}$, since the the black vertices dominate the graph $G$ and can only be \mbox{$2d$-dominated} by different vertices.

\vspace{0.5cm}

\begin{figure}[!htb]
\centering
\includegraphics[width=0.5\textwidth]{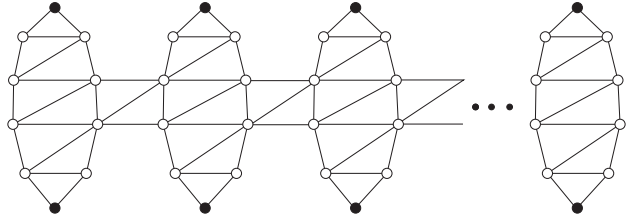}
\caption{A maximal outerplanar graph $G$ for which $\gamma_{2d}(G)=\frac{n}{5}$.} 
\label{FIG:article-arXiv-11}
\end{figure}

This example can be generalized to \mbox{$kd$-domination} to obtain \mbox{$\gamma_{kd}(n) \geq \frac{n}{(2k+1)}$}. For example, in Fig. \ref{FIG:article-arXiv-12}, the black vertices can only be \mbox{3-dominated} by different vertices, so \mbox{$\gamma_{3d}(n) \geq \frac{n}{7}$}.

\begin{figure}[!htb]
\centering
\includegraphics[width=0.5\textwidth]{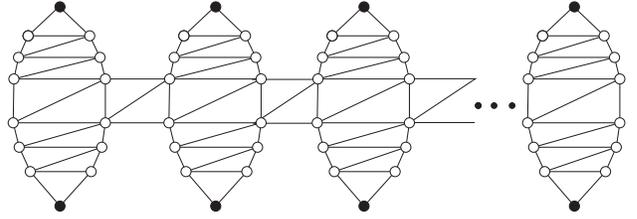}
\caption{A maximal outerplanar graph $G$ for which $\gamma_{3d}(G)=\frac{n}{7}$.} \label{FIG:article-arXiv-12}
\end{figure}

According to theorem \ref{Thm:inequalities}, \mbox{$\gamma_{2d}(G) \leq g_{2d}(G)$}, so \mbox{$\lfloor \frac{n}{5}\rfloor \leq g_{2d}(G)$}. In conclusion, $\lfloor \frac{n}{5}\rfloor$ $2d$-guards are occasionally necessary and always sufficient to guard a $n$-vertex maximal outerplanar graph $G$. On the other hand, we can also establish that \mbox{$\gamma_{2d} = \lfloor \frac{n}{5}\rfloor$}, since \mbox{$\lfloor \frac{n}{5}\rfloor \leq \gamma_{2d}(n)$} and \mbox{$\gamma_{2d}(n) \leq g_{2d}(n)$}, for all $n$. Thus, it follows:

\begin{theorem}
Every $n$-vertex maximal outerplanar graph with \mbox{$n \geq 5$} can be $2d$-guarded (and $2d$-dominated) by $\lfloor \frac{n}{5}\rfloor$ guards. This bound is tight in the worst case.
\end{theorem}

\section{$2d$-covering of maximal outerplanar graphs}
\label{SEC:CoveringTightUpperBounds}

In this section we determine an upper bound for $2d$-vertex cover on maximal outerplanar graphs and we show that this bound is tight. In order to do this, we first introduce the following lemma, whose proof is omitted, since it is analogous to the one of lemma \ref{Lem:f(m-1)Guarding}.

\begin{lemma}
\label{Lem:f(m-1)Covering}

Suppose that $f(m)$ vertices are always sufficient to $2d$-cover any outerplanar maximal graph with $m$ vertices. If $G$ is an arbitrary outerplanar maximal graph and if we choose any of its $m$ vertices to place a $2d$-covering vertex, then $f(m-1)$ additional vertices are sufficient $2d$-cover $G$.
\end{lemma}

\begin{theorem}
\label{Thm:SufficiencyCovering}
Every $n$-vertex maximal outerplanar graph, with \mbox{$n \geq 4$}, can be $2d$-covered with $\lfloor \frac{n}{4} \rfloor$ vertices. That is, $\beta_{2d}(n) \leq \lfloor \frac{n}{4} \rfloor$ for all \mbox{$n \geq 4$}.
\end{theorem}

\begin{proof}

For \mbox{$4 \leq n \leq 9$}, the truth of the theorem can be easily established -- the upper bounds are resumed in Table \ref{TAB:covering}. Note that for \mbox{$n=4$} the $2d$-covering vertex can be chosen randomly and for \mbox{$n=5$} it can be placed among the vertices of degree greater than 2.

\begin{table}[!htb]
\centering
\small
  \begin{tabular}{ | l | c | c | c | c | c | c | }
    \hline
    $n$ & 4  & 5  & 6  & 7  & 8  & 9  \\
    \hline \hline
    $\beta_{2d}(n)$ & 1  & 1  & 1  & 1  & 2  & 2  \\
    \hline
  \end{tabular}
  \caption{Number of  vertices that suffice to $2d$-cover a maximal outerplanar graph of $n$ vertices.}
  \label{TAB:covering}
\end{table}

Assume that \mbox{$n \geq 10$}, and that the theorem holds for \mbox{$n'<n$}. Lemma \ref{Lem:O'Rourke} guarantees the existence of an interior edge $e$ that partitions $G$ into maximal outerplanar graphs $G_1$ and $G_2$, where $G_1$ contains $m$ exterior edges of $G$ with $5 \leq m \leq 8$. Each value of $m$ will be considered separately.

\begin{enumerate}

\item[(1)] $m=5$. $G_1$ has $m+1=6$ exterior edges, thus $G_1$ can be $2d$-covered with one vertex. $G_2$ has $n-4$ exterior edges including $e$, and by the induction hypothesis, it can be $2d$-covered with $\lfloor \frac{n-4}{4} \rfloor = \lfloor \frac{n}{4} \rfloor -1$ vertices. Thus $G_1$ and $G_2$ together can be $2d$-covered by $\lfloor \frac{n}{4} \rfloor$ vertices.

\item[(2)] $m=6$. $G_1$ has $m+1=7$ exterior edges, thus $G_1$ can be $2d$-covered with one vertex. $G_2$ has $n-5$ exterior edges including $e$, and by induction hypothesis, it can be $2d$-covered with $\lfloor \frac{n-5}{4} \rfloor \leq \lfloor \frac{n}{4} \rfloor -1$ vertices. Thus $G_1$ and $G_2$ together can be $2d$-covered by $\lfloor \frac{n}{4} \rfloor$ vertices.

\item[(3)] $m=7$.  The presence of any of the internal edges (0,6), (0,5), (1,7) and (2,7) would violate the minimality of $m$. Thus, the triangle $T$ in $G_1$ that is bounded by $e$ is either (0,3,7) or (0,4,7). Since these are equivalent cases, suppose that $T$ is (0,3,7) as shown in Fig. \ref{FIG:article-arXix-13to14}(a). We distinguish two cases:

    \begin{enumerate}

        \item
         The degree of vertex 3 in the pentagon (3,4,5,6,7) is greater than 2. In this case vertex 3 is a vertex cover of $G_1$, and by induction hypothesis $G_2$ can be $2d$-covered with \mbox{$\lfloor \frac{n-6}{4} \rfloor \leq \lfloor \frac{n}{4} \rfloor-1$} vertices. Together with vertex 3 all of $G$ can be $2d$-covered by $\lfloor \frac{n}{4} \rfloor$ vertices.

        \item
         The degree of vertex 3 in the pentagon (3,4,5,6,7) is equal to 2. In this case vertex 7 $2d$-covers the pentagon (3,4,5,6,7). We consider graph $G^{*}$ that results from the union of $G_2$, triangle $T$ and quadrilateral (0,1,2,3). In this way, $G^{*}$ has \mbox{n-3} vertices. By lemma \ref{Lem:f(m-1)Covering} the vertex 7 permits the remainder of $G^{*}$ to be $2d$-covered by $f(n-3-1)$ additional vertices, where $f(n')$ is the number of vertices that are always sufficient to $2d$-cover a maximal outerplanar graph of $n'$ vertices. By the induction hypothesis $f(n')= \lfloor \frac{n'}{4} \rfloor$. Thus  $\lfloor \frac{n-4}{4} \rfloor = \lfloor \frac{n}{4} \rfloor - 1$ vertices are sufficient to $2d$-cover $G_1$. Together with vertex 7, all of $G$ is 2d-covered by $\lfloor \frac{n}{4}\rfloor$ vertices.

    \end{enumerate}

\begin{figure}[!htb]
  \centering
  \subfloat[]{\includegraphics[width=0.2\textwidth]{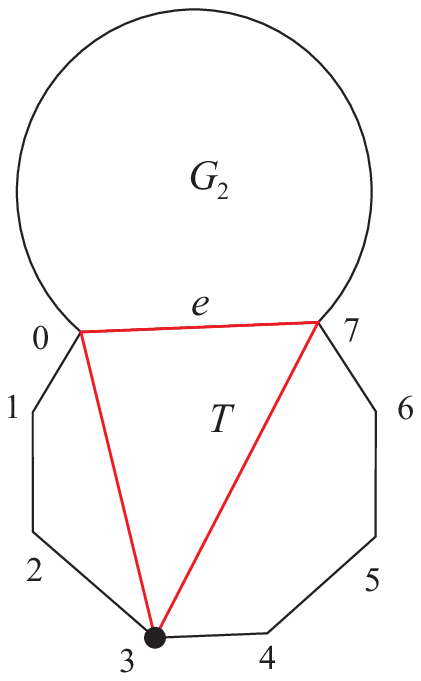}}
  \hspace{1cm}
  \subfloat[]{\includegraphics[width=0.2\textwidth]{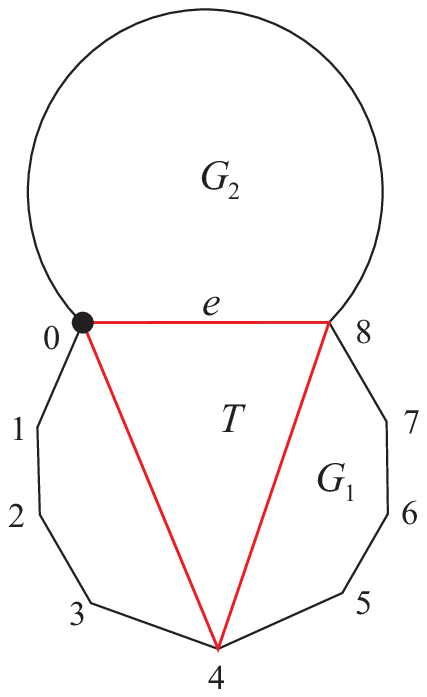}}
  \caption{The interior edge $e$ separates $G$ into two maximal outerplanar graphs $G_1$ and $G_2$: (a) $G_1$ has 8 exterior edges and the triangle $T$ in $G_1$ that is bounded by $e$ is (0,3,7); (b) $G_1$ has 9 exterior edges and the triangle $T$ in $G_1$ that is bounded by $e$ is (0,8,4).}
  \label{FIG:article-arXix-13to14}
  \label{FIG:article-arxiv-13to14}
\end{figure}

\item[(4)] $m=8$. $G_1$ has $m+1=9$ exterior edges, thus $G_1$ can be $2d$-covered with two vertices. We will consider two separate cases:

    \begin{enumerate}

        \item
        Vertices 0 and 8 have degree 2 in pentagons (0,1,2,3,4) and (4,5,6,7,8), respectively. Then vertex 5 $2d$-covers $G_1$. By the induction hypothesis $G_2$ can be $2d$-covered with \mbox{$\lfloor \frac{n-7}{4} \rfloor \leq \lfloor \frac{n}{4} \rfloor -1$} vertices. Thus $G$ can be $2d$-covered by $\lfloor \frac{n}{4} \rfloor$ vertices.

        \item
         Vertex 0 has degree greater than 2 in pentagon (0,1,2,3,4). In this case we place a guard at vertex 0 and  another guard in one vertex of the pentagon (4,5,6,7,8) whose degree is greater than 2. These two guards $2d$-cover $G_1$. $G_2$ has $n-7$ vertices. By lemma \ref{Lem:f(m-1)Covering} the  vertex 0 permits the remainder of $G_2$ to be 2d-covered by $f(n-7-1)=f(n-8)$ additional guards, where $f(n')$ is the number of $2d$-covering vertices that are always sufficient to $2d$-cover a maximal outerplanar graph with $n'$ vertices. By induction hypothesis $f(n')=\lfloor \frac{n'}{4} \rfloor$. Thus, $\lfloor \frac{n-8}{4} \rfloor = \lfloor \frac{n}{4} \rfloor-2$  vertices suffice to guard $G_2$. Together with the two allocated to $G_1$, all of $G$ is $2d$-covered by $\lfloor \frac{n}{4}\rfloor$ vertices.

\end{enumerate}

\end{enumerate}

\end{proof}

Now, we will prove that this upper bound is tight. The bold edges of the maximal outerplanar graph illustrated in Fig. \ref{FIG:article-arXiv-15} can only be $2d$-covered from different vertices, and therefore $\beta_{2d}(n) \geq \frac{n}{4}$.

\begin{figure}[!htb]
\centering
\includegraphics[width=0.5\textwidth]{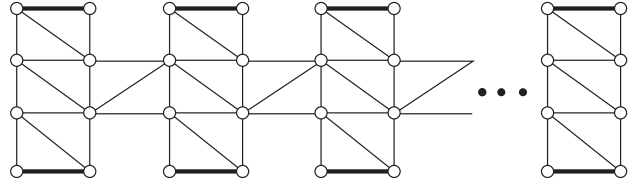}
\caption{A maximal outerplanar graph $G$ for which $\beta_{2d}(G)=\frac{n}{4}$.} \label{FIG:article-arXiv-15}
\end{figure}

As a conclusion,

\begin{theorem}
Every $n$-vertex maximal outerplanar graph with \mbox{$n \geq 5$} can be $2d$-covered by $\lfloor \frac{n}{4}\rfloor$ vertices. This bound is tight in the worst case.
\end{theorem}

\section{Conclusions and further research}
\label{SEC:Conclusions}

In this article we defined the concept of $kd$-guarding and formalized the distance $kd$-vertex cover. We showed that there is a relationship between $2d$-guarding, $2d$-dominating and $2d$-vertex cover sets on triangulation graphs. Furthermore, we proved tight bounds for \mbox{$n$-vertex} maximal outerplanar graphs: \mbox{$g_{2d}(n) = \gamma_{2d}(n) = \lfloor \frac{n}{5} \rfloor$} and \mbox{$\beta_{2d}(n)= \lfloor \frac{n}{4} \rfloor$}.

Regarding future research, we believe these bounds can be extended to any triangulation and are therefore not exclusive of maximal outerplanar graphs. Moreover, it would be interesting to study how these bounds evolve for $3d$-guarding, $3d$-dominating and $3d$-vertex cover sets. And of course study upper and lower bounds of the mentioned three concepts for any distance $k$. Finally, in the future we would like to study these distance concepts applied to other types of graphs, and not focus only on triangulations.






\end{document}